\documentclass[11pt]{article}
\usepackage[utf8]{inputenc}
\usepackage[T1]{fontenc}
\usepackage{fullpage}
\usepackage{amsthm}
\usepackage{amssymb,amsmath,mathtools}
\usepackage{subcaption}
\usepackage{paralist}
\usepackage{xspace}
\usepackage{color}
\usepackage[boxed,vlined]{algorithm2e}
\SetAlCapSkip{0.5em}
\usepackage[nocompress]{cite}
\usepackage{hyperref}
\usepackage{booktabs}
\usepackage{balance}
\usepackage{authblk}

\newtheorem{theorem}{Theorem}

\newcommand{\comm}[1]{}

\newcommand{\Oh}{{\ensuremath{\mathcal{O}}}}
\newcommand{\df}{\emph}

\newcommand{\Q}{\ensuremath{\mathcal{Q}}\xspace}
\newcommand{\D}{\ensuremath{\mathcal{D}}\xspace}
\DeclareMathOperator{\gap}{gap}
\DeclareMathOperator{\LogGap}{\ensuremath{LogGap}}
\DeclareMathOperator{\Log}{\ensuremath{Log}}
\newcommand{\BV}{\textsc{BV}\xspace}
\newcommand{\PEF}{\textsc{PEF}\xspace}
\newcommand{\BIC}{\textsc{BIC}\xspace}

\newcommand{\prob}[1]{\text{\textsc{#1}}\xspace}
\newcommand{\MLA}[0]{\prob{MLA}}
\newcommand{\MLogA}[0]{\prob{MLogA}}
\newcommand{\MLogGapA}[0]{\prob{MLogGapA}}
\newcommand{\BiMLogA}[0]{\prob{BiMLogA}}

\newcommand{\NP}[0]{\texttt{NP}}
\newcommand{\APX}[0]{\texttt{APX}}

\title{Compressing Graphs and Indexes with Recursive~Graph~Bisection}

\author[1]{Laxman~Dhulipala}
\author[2]{Igor~Kabiljo}
\author[2]{Brian~Karrer}
\author[2]{Giuseppe~Ottaviano}
\author[2]{Sergey~Pupyrev}
\author[2]{Alon~Shalita}
\affil[1]{Carnegie Mellon University}
\affil[2]{Facebook}

\begin{document}
\date{}

\maketitle
\begin{abstract}
  Graph reordering is a powerful technique to increase the locality of the
  representations of graphs, which can be helpful in several applications.
  We study how the technique can be
  used to improve compression of graphs and inverted indexes.

  We extend the recent theoretical model of Chierichetti~et~al. (KDD 2009) for
  graph compression, and show how it can be employed for compression-friendly
  reordering of social networks and web graphs and for assigning document
  identifiers in inverted indexes. We design and implement a novel theoretically
  sound reordering algorithm that is based on recursive graph bisection.

  Our experiments show a significant improvement of the compression rate of
  graph and indexes over existing heuristics. The new method is relatively
  simple and allows efficient parallel and distributed implementations,
  which is demonstrated on graphs with billions of vertices and hundreds of
  billions of edges.
\end{abstract}




\section{Introduction}

Many real-world systems and applications use in-memory representation of indexes
for serving adjacency information in a graph. A popular example is social networks
in which the list of friends is stored for every user. Another example is an inverted
index for a collection of documents that stores, for every term,
the list of documents where the term occurs. Maintaining these indexes requires a compact,
yet efficient, representation of graphs.

How to represent and compress such information? Many techniques for graph and index
compression have been studied in the literature~\cite{WMB99,MP15}.
Most techniques first sort vertex identifiers in an adjacency list, and then replace the identifiers
(except the first) with differences between consecutive ones. The
resulting gaps are encoded using some integer compression algorithm. Note that
using gaps instead of original identifiers decreases the values needed to be
compressed and results in a higher compression ratio. We stress that
the success of applying a particular encoding algorithm strongly depends
on the distribution of gaps in an adjacency list: a sequence of small and regular gaps
is more compressible than a sequence of large and random ones.

This observation has motivated the approach of assigning identifiers
in a way that optimizes compression. \df{Graph reordering} has been
successfully applied for social networks~\cite{CKLMPR09,BRSV11}. In that
scenario, placing similar social actors nearby in the resulting order
yields a significant compression improvement. Similarly, lexicographic locality
is utilized for compressing the Web graph: when pages are ordered by URL,
proximal pages have similar sets of neighbors, which results in an increased
compression ratio of the graph, when compared with the compression obtained
using the original graph~\cite{RSWW02,BV04}.
In the context of index compression, the corresponding approach is called the
\df{document identifier assignment} problem. Prior work shows that for
many collections, index compression can be significantly improved
by assigning close identifiers to similar documents~\cite{SCSC03,Sil07,DAS10,BB05,BB02}.

In this paper, we study the problem of finding the best ``compression-friendly''
order for a graph or an inverted index. While graph reordering and document identifier
assignment are often studied independently, we propose a unified model that
generalizes both of the problems. Although a number of heuristics for the problems
exists, 
none of them provides any guarantees on the
resulting quality. In contrast, our algorithm is inspired by a theoretical approach
with provable guarantees on the final quality, and it is designed to
directly optimize the resulting compression ratio.
Our main contributions are the following.

\begin{compactitem}
    \item We analyze and extend the formal model of graph compression suggested in~\cite{CKLMPR09};
    the new model is suitable for both
    graph reordering and document identifier assignment problems.
    We show that the underlying optimization problem is \NP-hard (thus, resolving an
    open question stated in~\cite{CKLMPR09}), and suggest
    an efficient approach for solving the problem, which is based on approximation
    algorithms for graph reordering.

    \item Based on the theoretical result, we develop a practical algorithm for constructing
    compression-friendly vertex orders. The algorithm uses recursive graph
    bisection as a subroutine and tries to optimize a desired objective at every recursion step.
    Our objective corresponds to the size of the graph compressed using delta-encoding.
    The algorithm is surprisingly simple, which allows for efficient parallel and
    distributed implementations.

    \item Finally, we perform an extensive set of experiments on a collection of large real-world graphs,
    including social networks, web graphs, and search indexes. The experiments
    indicate that our new method outperforms the state-of-the-art graph reordering techniques,
    improving the resulting compression ratio. Our implementation is highly scalable
    and is able to process a billion-vertex graph in a few hours.
\end{compactitem}


The paper is organized as follows. We first discuss existing approaches for
graph reordering, assigning document identifiers, and the most popular
encoding schemes for graph and index representation (Section~\ref{sect:related}).
Then we consider algorithmic aspects of the underlying optimization problem.
We analyze the models for graph compression suggested by Chierichetti~et~al.~\cite{CKLMPR09}
and suggest our generalization in Section~\ref{sect:model}. Next, in Section~\ref{sect:approx},
we examine existing theoretical techniques for the graph reordering problem and
use the ideas to design a practical algorithm. A detailed description of the algorithm
along with the implementation details is presented in Section~\ref{sect:alg}, which
is followed by experimental Section~\ref{sect:exp}. We conclude the paper with the most
promising future directions in Section~\ref{sect:future}.

\section{Related Work}
\label{sect:related}
There exists a rich literature on graph and index compression, that can be roughly
divided into three categories: (1)~structural approaches that find and merge repeating
graph patterns (e.g., cliques), (2)~encoding adjacency data represented by a list of integers
given some vertex/document order, and (3)~finding a suitable order of graph vertices.
Our focus is on the ordering techniques.
We discuss the existing approaches for graph reordering, followed by
an overview of techniques for document identifier assignment.
Since many integer encoding algorithms can benefit from
such a reordering, we also outline the most popular encoding schemes.

\paragraph{Graph Reordering}
Among the first approaches for compressing large-scale graphs is a work by Boldi
and Vigna~\cite{BV04}, who compress web graphs using a \df{lexicographical
  order} of the URLs. Their compression method relies on two properties:
locality (most links lead to pages within the same host) and similarity (pages
on the same host often share the same links).
Later Apostolico and Drovandi~\cite{AD09} suggest one of the first ways to compress a
graph assuming no a priori knowledge of the graph. 
The technique is based on a \df{breadth-first traversal} of the graph
vertices and achieves a better compression rate using an entropy-based encoding.

Chierichetti et al.~\cite{CKLMPR09} consider the theoretical aspect of the reordering
problem motivated by compressing social networks.
They develop a simple but practical heuristic for the problem, called \df{shingle ordering}.
The heuristic is based on obtaining a fingerprint of the neighbors of a vertex and
positioning vertices with identical fingerprints close to each other.
If the fingerprint can capture locality and similarity of the vertices,
then it can be effective for compression.
This approach is also called \df{minwise hashing} and was originally applied
by Broder~\cite{Bro97} for finding duplicate web pages.

Boldi et al.~\cite{BRSV11} suggest a reordering algorithm, called \df{Layered Label
Propagation}, to compress social networks. The algorithm is built on
a scalable graph clustering technique by label propagation~\cite{RAK07}. The idea
is to assign a label for every vertex of a graph based on the labels of its
neighbors. The process is executed in rounds until no more updates take place. Since
the standard label propagation described in~\cite{RAK07} tends to produce a giant
cluster, the authors of~\cite{BRSV11} construct a hierarchy of clusters. The vertices
of the same cluster are then placed together in the final order.

The three-step \df{multiscale} paradigm is often employed for the graph ordering problems. 
First, a sequence of coarsened graphs,
each approximating the original graph but having a smaller size, is created.
Then the problem is solved on the coarsest level by an exhaustive search. Finally, the process is reverted by
an uncoarsening procedure so that a solution for every graph in the sequence is based on
the solution for a previous smaller graph.
Safro and Temkin~\cite{ST11} employ the
algebraic multigrid methodology in which the sequence of coarsened graphs is constructed
using a projection of graph Laplacians into a lower-dimensional space.

\df{Spectral} methods have also been successfully applied to
graph ordering problems~\cite{JM92}.
Sequencing the vertices is done by sorting them according to
corresponding elements of the second smallest eigenvector of graph
Laplacian (also called the Fiedler vector).
It is known that the order yields the best non-trivial solution to a
relaxation of the \df{quadratic graph ordering problem}, and hence, it is a good heuristic for
computing linear arrangements.

Recently Kang and Faloutsos~\cite{YKF14} present another technique, called
\df{SlashBurn}. Their method constructs a permutation of graph vertices so that
its adjacency matrix is comprised of a few nonzero blocks. Such dense blocks are
easier to encode, which is beneficial for compression.

In our experiments, we compare our new algorithm with all
of the methods, which are either easy to implement, or come with the
source code provided by the authors.

\paragraph{Document Identifier Assignment}
Several papers study how to assign document identifiers in a document collection
for better compression of an inverted index. A popular idea is to perform a
clustering on the collection and assign close identifiers to similar
documents. Shieh et al.~\cite{SCSC03} propose a reassignment heuristic
motivated by the maximum \df{travelling salesman problem} (TSP). The heuristic computes
a pairwise similarity between every pairs of documents (proportional to the
number of shared terms), and then finds the longest path traversing the documents in
the graph.  An alternative algorithm calculating cosine similarities between
documents is suggested by Blandford and Blelloch~\cite{BB02}. Both methods are
computationally expensive and are limited to fairly small datasets.  The
similarity-based approach is later improved by Blanco and Barreiro~\cite{BB05}
and by Ding et al.~\cite{DAS10}, who make it scalable by reducing the size of
the similarity graph, respectively through dimensionality reduction and locality
sensitive hashing.

The approach by Silvestri~\cite{Sil07} simply sorts the collection of web pages
by their URLs and then assigns document identifiers according to the
order. The method performs very well in practice and is highly scalable.
This technique however does not generalize to document collections that do not have URL-like
identifiers.

\paragraph{Encoding Schemes}%
Our algorithm is not tailored specifically for an encoding scheme; any method
that can take advantage of lists with higher local density (clustering)
should benefit from the reordering. For our experiment we choose a few encoding
schemes that should be representative of the state-of-the-art.

Most graph compression schemes build on \emph{delta-encoding}, that is, sorting
the adjacency lists (\df{posting lists} in the inverted indexes case) so that
the gaps between consecutive elements are positive, and then encoding these gaps
using a variable-length integer code.
The WebGraph framework adds the ability to copy portions of the adjacency lists
from other vertices, and has special cases for runs of consecutive
integers. Introduced in 2004 by Boldi and Vigna~\cite{BV04}, it is still widely
used to compress web graphs and social networks.

Inverted indexes are usually compressed with more specialized techniques in
order to enable fast skipping, which enables efficient list intersection.  We
perform our experiments with Partitioned Elias-Fano and Binary Interpolative
Coding.
The former, introduced by Ottaviano and Venturini~\cite{OV14}, provides one of
the best compromise between decoding speed and compression ratio. The latter,
introduced by Moffat and Stuiver~\cite{moffat00binary}, has the highest
compression ratio in the literature, with the trade-off of slower decoding by
several times. Both techniques directly encode monotone lists, without going
through delta-encoding.

\section{Algorithmic Aspects}
\label{sect:theory}
Graph reordering is a combinatorial optimization problem with a goal to
find a linear layout of an input graph so that a certain objective function
(referred to as a \df{cost function} or just \df{cost})
is optimized. A linear layout of a graph $G=(V,E)$ with $n=|V|$ vertices
is a bijection $\pi: V\rightarrow \{1, \dots, n\}$. A layout is also
called an order, an arrangement, or a numbering
of the vertices.
In practice it is desirable that ``similar'' vertices of the graph are
``close'' in $\pi$. This leads to a number of problems that
we define next.

The \df{minimum linear arrangement} (\MLA) problem is to find a layout $\pi$
so that $$\sum_{(u,v) \in E} |\pi(u) - \pi(v)|$$ is minimized. This is
a classical \NP-hard problem~\cite{GJ79}, even when restricted to certain
graph classes.
The problem is \APX-hard under Unique Games Conjecture~\cite{DKSV06},
that is, it is unlikely that an efficient approximation algorithm exists.
Charikar et al.~\cite{CHKR10} suggested the best currently
known algorithm with approximation factor $\Oh(\sqrt{\log n} \log \log n)$;
see~\cite{Pet13} for a survey of results on \MLA.


A closely related problem is \df{minimum logarithmic arrangement} (\MLogA)
in which the goal is to minimize
$$\sum_{(u,v) \in E} \log |\pi(u) - \pi(v)|.$$
Here and in the following we denote $\log(x) = 1 + \lfloor \log_2(x) \rfloor$, that is,
the number of bits needed to represent an integer~$x$.
The problem is also \NP-hard, and one can show that the
optimal solutions of \MLA and \MLogA are quite different on some graphs~\cite{CKLMPR09}.
In practice a graph is represented in memory as an adjacency list using an encoding scheme;
hence, the gaps induced by consecutive neighbors of a vertex are important for compression.
For this reason, the \df{minimum logarithmic
gap arrangement} (\MLogGapA) problem is introduced. For a vertex $v\in V$
of degree $k$ and an order $\pi$, consider the neighbors $out(v) = (v_1, \dots, v_k)$
of $v$ such that $\pi(v_1) < \dots < \pi(v_k)$. Then the cost
compressing the list $out(v)$ under $\pi$ is related to
$f_{\pi} (v, out(v)) = \sum_{i=1}^{k-1} \log |\pi(v_{i+1}) - \pi(v_i)|$.
\MLogGapA consists in finding an order $\pi$, which minimizes
$$\sum_{v\in V} f_{\pi} (v, out(v)).$$
To the best of our knowledge, \MLogA and \MLogGapA are introduced quite recently
by Chierichetti et al.~\cite{CKLMPR09}. They show that \MLogA is \NP-hard but
left the computational complexity of \MLogGapA open.
Since the latter problem is arguably more important for applications,
we address the open question of complexity of the problem.

\begin{theorem}
\label{thm:hard}
    \MLogGapA is \NP-hard.
\end{theorem}

\begin{proof}
    We prove the theorem by using the hardness of \MLogA, which is known to be \NP-hard~\cite{CKLMPR09}.
    Let $G=(V,E)$ be an instance of \MLogA. We build a bipartite graph $G'=(V',E')$ by splitting
    every edge of $E$ by a degree-2 vertex. Formally, we add $|E|$ new vertices so that $V'=V \cup U$,
    where $V=\{v_1,\dots,v_n\}$ and $U=\{u_1,\dots,u_m\}$. For every edge $(a,b)\in E$, we have
    two edges in $E'$, that is, $(a,u_i)$ and $(b,u_i)$ for some $1 \le i \le m$.
    Next we show that an optimal solution for \MLogGapA on $G'$ yields an optimal solution for \MLogA on $G$,
    which proves the claim of the theorem.

    Let $R$ be an optimal order of $V'$ for \MLogGapA. Observe that without loss of generality, the vertices of
    $V$ and $U$ are separated
    in $R$, that is, $R=(v_{i_1}, \dots, v_{i_n}, u_{j_1}, \dots, u_{j_m})$.
    Otherwise, the vertices can be reordered so that the total objective is not increased.
    To this end, we ``move'' all the vertices of $V$ to the left of $R$ by preserving their relative order.
    It is easy to see that the gaps between vertices of $V$ and the gaps between vertices of $U$ can only
    decrease.

    Now the cost of \MLogGapA is $$\sum_{v\in V} f_{\pi_u} (v, out(v)) + \sum_{u\in U} f_{\pi_v} (u, out(u)),$$
    where $\pi_u=(u_{j_1}, \dots, u_{j_m})$ and $\pi_v=(v_{i_1}, \dots, v_{i_n})$.
    Notice that the second term of the sum depends only on the order $\pi_v$ of the vertices in $V$, and
    it equals to the cost of \MLogA for graph $G$. Since $R$ is optimal for \MLogGapA, the order
    $\pi_v=(v_{i_1}, \dots, v_{i_n})$ is also optimal for \MLogA.
\end{proof}

Most of the previous works consider the \MLogA problem for graph compression, and
the algorithms are not directly suitable for index compression. Contrarily,
an inverted index is generally represented by a directed graph (with edges from terms to documents),
which is not captured by the \MLogGapA problem. In the following, we suggest a model, which
generalizes both \MLogA and \MLogGapA and better expresses graph and
index compression.

\subsection{Model for Graph and Index Compression}
\label{sect:model}
Intuitively, our new model is a bipartite graph comprising
of \df{query} and \df{data} vertices. A query vertex might correspond to an actor in a
social network or to a term in an inverted index. Data vertices are an actor's friends
or documents containing the term, respectively. The goal is to find a layout of data vertices.

Formally, let $G=(\Q \cup \D, E)$ be an undirected unweighted bipartite
graph with disjoint sets of vertices $\Q$ and $\D$. We denote $|\D|=n$ and $|E|=m$.
The goal is to find a permutation, $\pi$, of data vertices, $\D$, so that the following objective
is minimized:
$$
\sum_{q \in \Q} \sum_{i=1}^{\deg_q-1} \log (\pi(u_{i+1}) - \pi(u_i)),
$$
where $\deg_q$ is the degree of query vertex $q \in \Q$, and $q$'s neighbors are
$\{u_1, \dots, u_{deg_q}\}$ with $\pi(u_1) < \dots < \pi(u_{deg_q})$.
Note that the objective is closely related to minimizing the number of bits needed
to store a graph or an index represented using the delta-encoding scheme. We call
the optimization problem \df{bipartite minimum logarithmic arrangement} (\BiMLogA),
and the corresponding cost averaged over the number of gaps $\LogGap$.

Note that \BiMLogA is different from \MLogGapA in that the latter does not differentiate
between data and query vertices (that is, every vertex is query and data in \MLogGapA),
which is unrealistic in some applications.
It is easy to see that the new problem generalizes both \MLogA and \MLogGapA: to model
\MLogA, we add a query vertex for every edge of the input graph, as in the proof of
Theorem~\ref{thm:hard}; to model \MLogGapA, we add a query for every vertex of the
input graph; see Figure~\ref{fig:reduction}. Moreover, the new approach can be naturally
applied for compressing directed graphs; to this end, we only consider gaps induced by outgoing
edges of a vertex.
Clearly, given an algorithm for \BiMLogA, we can easily solve
both \MLogA and \MLogGapA. Therefore, we focus on this new problem in the next sections.


\begin{figure}[htbp]
    \begin{minipage}[b]{0.98\textwidth}
        \begin{subfigure}[t]{.31\textwidth}
            \centering
            \includegraphics[page=1]{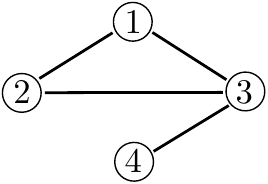}
            \caption{Original graph}
        \end{subfigure}
        \hfill
        \begin{subfigure}[t]{.31\textwidth}
            \centering
            \includegraphics[page=2]{pics/reduction}
            \caption{\MLogA}
        \end{subfigure}
        \hfill
        \begin{subfigure}[t]{.31\textwidth}
            \centering
            \includegraphics[page=3]{pics/reduction}
            \caption{\MLogGapA}
        \end{subfigure}
        \caption{Modeling of \MLogA and \MLogGapA with a bipartite graph with
            query (red) and data (blue) vertices.}
        \label{fig:reduction}
    \end{minipage}
\end{figure}

How can one solve the above ordering problems? Next we discuss the existing
theoretical approaches for solving graph ordering problems. We focus
on approximation algorithms, that is, efficient algorithms for \NP-hard problems
that produce sub-optimal solutions with provable quality.

\subsection{Approximation Algorithms}
\label{sect:approx}

To the best of our knowledge, no approximation algorithms exist for the new variants of the graph ordering
problem. However, a simple observation shows that every algorithm has approximation factor $\Oh(\log n)$.
Note that the cost of a gap between $u \in \D$ and $v \in \D$ in \BiMLogA cannot exceed $\log n$, as $|\pi(v)-\pi(u)| \le n$
for every permutation $\pi$. On the other hand, the cost of a gap is at least $1$. Therefore,
an arbitrary order of vertices yields a solution with the cost, which is at most $\log n$ times greater
than the optimum.

In contrast, the well-studied \MLA does not admit such a simple approximation and requires more involved
algorithms. We observe that most of the existing algorithms for \MLA and related ordering problems
employ the divide-and-conquer approach; see Algorithm~\ref{algo:order}.
Such algorithms partition the vertex set into two sets of roughly equal size, compute recursively an
order of each part, and ``glue'' the orderings of the parts together.
The crucial part is \df{graph bisection} or more generally \df{balanced graph partitioning},
if the graph is split into more than two parts.

\begin{algorithm}[h]
    \caption{Graph Reordering using Graph Bisection}
    \label{algo:order}

    \SetKwInOut{Input}{Input}
    \Input{graph $G$}
    \BlankLine
    1. Find a bisection $(G_1, G_2)$ of $G$\;
    2. Recursively find linear arrangements for $G_1$ and $G_2$\;
    3. Concatenate the resulting orderings\;
\end{algorithm}

The first non-trivial approximation algorithm for \MLA follows the above approach. Hancen~\cite{Han89}
proves that Algorithm~\ref{algo:order} yields an
$\Oh(\alpha \log n)$-approximation for \MLA, where $\alpha$ indicates how
close is the solution of the first step (bisection of $G$) to the optimum.
Later, Charikar et al.~\cite{CMM07} shows that a tighter analysis is possible, and the algorithm is in fact
$\Oh(\alpha)$-approximation for $d$-dimensional \MLA. Currently, $\alpha=\Oh(\sqrt{\log n})$
is the best known bound~\cite{ARV09}.
Subsequently, the idea of Algorithm~\ref{algo:order} was employed by Even et al.~\cite{ENRS00},
Rao and Richa~\cite{RR98}, and Charikar et al.~\cite{CHKR10} for composing approximation
algorithms for \MLA. The techniques use the recursive
divide-and-conquer approach and utilize a spreading metric by solving a linear program with an
exponential number of constraints.

Inspired by the algorithms, we design a practical approach for the \BiMLogA problem.
While solving a linear program is not feasible for large graphs, we utilize recursive
graph partitioning in designing the algorithm. Next we describe all the steps
and provide implementation-specific details.

\section{Compression-Friendly Graph Reordering}
\label{sect:alg}

Assume that the input is an undirected bipartite graph $G=(\Q \cup \D, E)$,
and the goal is to compute an order of $\D$.
On a high level, our algorithm is quite simple; see Algorithm~\ref{algo:order}.

The reordering method is based on the graph bisection
problem, which asks for a partition of graph vertices into two sets of
equal cardinality so as to minimize an objective function. Given an input
graph $G$ with $|\D|=n$, we apply the bisection algorithm to obtain
two disjoint sets
$V_1,V_2 \subseteq \D$ with $|V_1|=\lfloor n/2 \rfloor$ and $|V_2|=\lceil n/2 \rceil$.
We shall lay out $V_1$ on the set $\{1, \dots, \lfloor n/2 \rfloor\}$ and
lay out $V_2$ on the set $\{\lceil n/2 \rceil, \dots, n\}$.
Thus, we have divided the problem into two problems
of half the size, and we recursively compute good layouts for the graphs induced by $V_1$
and $V_2$, which we call $G_1$ and $G_2$, respectively. Of course, when there
is only one vertex in $G$, the order is trivial.

How to bisect the vertices of the graph? We use a graph bisection method, similar to the
popular Kernighan-Lin heuristic~\cite{KL70}; see Algorithm~\ref{algo:bp}.
Initially we split $\D$ into two sets, $V_1$ and $V_2$, and define a computational cost of the partition,
which indicates how ``compression-friendly'' the partition is. Next we exchange pairs
of vertices in $V_1$ and $V_2$ trying to improve the cost. To this end
we compute, for every vertex $v \in \D$, the \df{move gain}, that is, the
difference of the cost after moving $v$ from its current set to another one.
Then the vertices of $V_1$ ($V_2$) are sorted
in the decreasing order of the gains to produce list $S_1$ ($S_2$). Finally,
we traverse the lists $S_1$ and $S_2$ in the order and exchange
the pairs of vertices, if the sum of their move gains is positive.
Note that unlike classical graph bisection heuristics~\cite{KL70,FM82},
we do not update move gains after every swap.
The process is repeated until the convergence criterion is met (no swapped vertices)
or the maximum number of iterations is reached.

\begin{algorithm}[h]
    \caption{Graph Bisection}
    \label{algo:bp}

    \SetKwInOut{Input}{Input}
    \SetKwInOut{Output}{Output}
    \Input{graph $G=(\Q \cup \D, E)$}
    \Output{graphs $G_1=(\Q \cup V_1, E_1),G_2=(\Q \cup V_2, E_2)$}
    \BlankLine
    determine an initial partition of $\D$ into $V_1$ and $V_2$\;
    \Repeat{converged {\bf or} iteration limit exceeded}{
        \For{$v \in \D$}{$gains[v] \leftarrow ComputeMoveGain(v)$}
        $S_1 \leftarrow $ sorted $V_1$ in descending order of $gains$\;
        $S_2 \leftarrow $ sorted $V_2$ in descending order of $gains$\;
        \For{$v \in S_1$, $u \in S_2$}{
            \If{$gains[v]+gains[u] > 0$}{
                exchange $v$ and $u$ in the sets\;}
            \lElse{break}
        }
    }
    \Return{graphs induced by $\Q \cup V_1$ and $\Q \cup V_2$}
\end{algorithm}

To initialize the bisection, we consider the following two alternatives.
A simpler one is to arbitrarily split $\D$ into two equal-sized sets.
Another approach is based on shingle ordering (minwise hashing) suggested in~\cite{CKLMPR09}.
To this end, we order the vertices as described in~\cite{CKLMPR09} and assign
the first $\lfloor n/2 \rfloor$ vertices to $V_1$ and the last $\lceil n/2 \rceil$
to $V_2$.

Algorithm~\ref{algo:bp} tries to minimize the following objective function of the
sets $V_1$ and $V_2$, which is motivated by \BiMLogA.
For every vertex
$q \in \Q$, let $\deg_1(q) = |\{(q, v): v \in V_1\}|$, that is, the number of
adjacent vertices in set $V_1$; define $\deg_2(q)$ similarly. Then
the cost of the partition is
$$
\sum_{q \in \Q}\left( \deg_1(q) \log(\frac{n_1}{\deg_1(q) + 1}) +
\deg_2(q) \log(\frac{n_2}{\deg_2(q) + 1}) \right),
$$
where $n_1=|V_1|$ and $n_2=|V_2|$.
The cost estimates the required number of bits needed to represent $G$ using delta-encoding.
If the neighbors of $q \in \Q$ are uniformly distributed in
the final arrangement of $V_1$ and $V_2$, then the the average gap between consecutive numbers in the $q$'s
adjacency list is $\gap_1 \coloneqq n_1 / (\deg_1(q) + 1)$ and
$\gap_2 \coloneqq n_2 / (\deg_2(q) + 1)$ for $V_1$ and $V_2$, respectively; see Figure~\ref{fig:q-gap}.
There are $(\deg_1(q) - 1)$ gaps between vertices in $V_1$ and $(\deg_2(q) - 1)$
gaps between vertices in $V_2$.  Hence, we need approximately
$(\deg_1(q) - 1) \log(\gap_1) + (\deg_2(q) - 1) \log(\gap_2)$
bits to compress the within-group gaps.  In addition, we have to account for the average
gap between the last vertex of $V_1$ and the first vertex of $V_2$, which is
$(\gap_1 + \gap_2)$.
Assuming that $n_1 = n_2$, we have
$\log(\gap_1 + \gap_2) = \log(\gap_1) + \log(\gap_2)+ C$,
where $C$ is a constant with respect to the data vertex assignment, and hence,
it can be ignored in the optimization.
Adding this between-group contribution to the within-group contributions gives the above expression.

\begin{figure}[!h]
    \centering
    \includegraphics[width=0.24\textwidth]{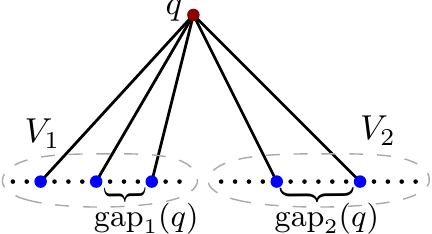}
    \caption{Partitioning $\D$ into $V_1$ and $V_2$
        for a query $q \in \Q$ with $\deg_1(q)=3$ and $\deg_2(q)=2$.}
    \label{fig:q-gap}
\end{figure}

Note that using the cost function, it is straightforward to implement
$ComputeMoveGain(v)$ function from Algorithm~\ref{algo:bp} by traversing all
the edges $(q, v) \in E$ for $v\in \D$ and summing up the cost differences of
moving $v$ to another set.

Combining all the steps of Algorithm~\ref{algo:order} and Algorithm~\ref{algo:bp},
we have the following claim.

\begin{theorem}
    The algorithm produces a vertex order in $\Oh(m \log n + n \log^2 n)$ time.
\end{theorem}

\begin{proof}
There are $\lceil \log n \rceil$ levels of recursion. Each call of graph bisection
requires computing move gains and sorting of $n$ elements. The former can be done
in $\Oh(m)$ steps, while the latter requires $\Oh(n \log n)$ steps. Summing over all
subproblems, we get the claim of the theorem.
\end{proof}

\subsection{Implementation}
Due to the simplicity of the algorithm, it can be efficiently implemented in parallel
or distributed manner. For the former, we notice that two different recursive calls
of Algorithm~\ref{algo:order} are independent, and thus, can be executed in parallel.
Analogously, a single bisection procedure can easily be parallelized, as
each of its steps computes independent values for every vertex, and a parallel
implementation of sorting can be used.
In our implementation, we employ the fork-join computation model in which small enough
graphs are processed sequentially, while larger graphs which occur on the first few
levels of recursion are solved in parallel manner.

Our distributed implementation relies on the vertex-centric programming model
and runs in the Giraph framework\footnote{http://giraph.apache.org}. 
In Giraph, a computation is split into supersteps that consists
of processing steps: (i)~a vertex executes a user-defined function based on
local vertex data and on data from adjacent vertices, (ii)~the resulting output
is sent along outgoing edges. Supersteps end with a synchronization barrier, which
guarantees that messages sent in a given superstep are received at the beginning
of the next superstep. The whole computation is executed iteratively for a
certain number of rounds, or until a convergence property is met.

Algorithm~\ref{algo:bp} is implemented in the vertex-centric model with
a simple modification. The first two supersteps compute move gains for all data vertices.
To this end, every query vertex calculates the differences of the cost function
when its neighbor moves from a set to another one. Then, every data vertex
sums up the differences over its query neighbors.
Given the move gains, we exchange the vertices as follows. Instead of sorting the move gains,
we construct, for both sets, an approximate histogram (e.g., as described in~\cite{YE10})
of the gain values.
Since the size of the histograms is small enough, we collect the data on a dedicated
host, and decide how many vertices from each bin should exchange its set. On the
last superstep, this information is propagated over all data vertices and the corresponding
swaps take effect.

\section{Experiments}
\label{sect:exp}
We design our experiments to answer two primary questions:
(i)~How well does our algorithm compress graphs and indexes in comparison with
existing techniques? (ii)~How do various parameters of the algorithm
contribute to the solution, and what are the best parameters?

\subsection{Dataset}
For our experiments, we use several publicly available web graphs,
social networks, and inverted document indexes; see Table~\ref{table:dataset}.
In addition, we run evaluation on
two large subgraphs of the Facebook friendship graph and a sample of the Facebook search index.
These private datasets serve to demonstrate scalability of our approach.
We do not release the datasets and our source code due to corporate restrictions.
Before running the tests, all the graphs are made unweighted and converted to bipartite
graphs as described in Section~\ref{sect:model}. Our dataset is as follows.

\begin{compactitem}
    \item \texttt{Enron} represents an email communication network;
    data is available at \url{https://snap.stanford.edu/data}.

    \item \texttt{AS-Oregon} is an Autonomous Systems peering information inferred from Oregon
    route-views in 2001; data is available at \url{https://snap.stanford.edu/data}.

    \item \texttt{FB-NewOrlean} contains a list of all of the user-to-user links from the Facebook
    New Orleans network; the data was crawled and anonymized in 2009~\cite{VMCG09}.

    \item \texttt{web-Google} represents web pages with hyperlinks between them.
    The data was released in 2002 by Google; data is available at \url{https://snap.stanford.edu/data}.

    \item \texttt{LiveJournal} is an undirected version of the public social graph (snapshot from 2006) containing
    $4.8$ million vertices and $42.9$ million edges~\cite{UB13}.

    \item \texttt{Twitter} is a public graph of tweets, with about $41$ million vertices (twitter accounts) and $2.4$
    billion edges (denoting followership)~\cite{KCHM10}.

  \item \texttt{Gov2} is an inverted index built on the TREC 2004 Terabyte Track
    test collection, consisting of 25 million .gov sites crawled in early 2004.

  \item \texttt{ClueWeb09} is an inverted index built on the ClueWeb 2009 TREC
    Category B test collection, consisting of 50 million English web pages
    crawled in 2009.

  \item \texttt{FB-Posts-1B} is an inverted index built on a sample of one billion
    Facebook posts, containing the longest posting lists. Since the
    posts have no hierarchical URLs, the
    \texttt{Natural} order for this index is random.

  \item \texttt{FB-300M} and \texttt{FB-1B} are two subgraphs of the Facebook friendship graph;
  the data was anonymized before processing.

\end{compactitem}

To build the inverted indexes for \texttt{Gov2} and \texttt{ClueWeb09} the body
text was extracted using Apache Tika\footnote{http://tika.apache.org} and the
words were lowercased and stemmed using the Porter2 stemmer; no stopwords were
removed. We consider only long posting lists containing more than $4096$ elements.

\begin{table}[h]
\vspace{-0.5cm}
\small
    \centering
    \begin{tabular}{lrrr}\toprule
        \centering Graph & \multicolumn{1}{c}{$|\Q|$} & \multicolumn{1}{c}{$|\D|$} & \multicolumn{1}{c}{$|E|$} \\
        \midrule
        \texttt{Enron}	  	  & $9,\!660$   & $9,\!660$   & $224,\!896$ \\
        \texttt{AS-Oregon}	  & $13,\!579$  & $13,\!579$   & $74,\!896$ \\
        \texttt{FB-NewOrlean} & $63,\!392$  & $63,\!392$   & $1,\!633,\!662$ \\
        \texttt{web-Google}	  & $356,\!648$  & $356,\!648$   & $5,\!186,\!648$ \\
        \texttt{LiveJournal}  & $4,\!847,\!571$  & $4,\!847,\!571$   & $85,\!702,\!474$ \\
        \texttt{Twitter}	  & $41,\!652,\!230$  & $41,\!652,\!230$   & $2,\!405,\!026,\!092$ \\
        \texttt{Gov2}	   	  & $39,\!187$  & $24,\!618,\!755$ & $5,\!322,\!924,\!226$ \\
        \texttt{ClueWeb09}	  & $96,\!741$  & $50,\!130,\!884$   & $14,\!858,\!911,\!083$ \\
        \texttt{FB-Posts-1B}  & $60 \times 10^3$ & $1 \times 10^9$ & $20 \times 10^9$ \\
        \texttt{FB-300M}	  & $300 \times 10^6$  	& $300 \times 10^6$	 & $90 \times 10^9$ \\
        \texttt{FB-1B}	  	  & $1 \times 10^9$ & $1 \times 10^9$ & $300 \times 10^9$ \\
        \bottomrule
    \end{tabular}
    \caption{Basic properties of our dataset.}
    \label{table:dataset}
\end{table}

\subsection{Techniques}
We compare our new algorithm (referred to as \texttt{BP}) with the following competitors.

\begin{compactitem}
    \item \texttt{Natural} is the most basic order defined for a graph. For
    web graphs and document indexes, the order is the URL lexicographic ordering
    used in~\cite{RSWW02,BV04}.
    For social networks, the order is induced by the original adjacency matrix.

    \item \texttt{BFS} is given by the bread-first search graph traversal algorithm as utilized
    in~\cite{AD09}.

  \item \texttt{Minhash} is the lexicographic order of $10$ minwise hashes of
    the adjacency sets. The same approach with only $2$ hashes is called
    \df{double shingle} in~\cite{CKLMPR09}.

    \item \texttt{TSP} is a heuristic for document reordering suggested by
    Shieh et al.~\cite{SCSC03}, which is based on solving the maximum travelling
    salesman problem.
    We implemented the variant of the algorithm that performs
    best in the authors' experiments. Since the algorithm is computationally expensive, we run
    it on small instances only. The sparsification techniques presented in~\cite{BB05,DAS10} would
    allow us to scale to the larger graphs, but they are too complex to re-implement
    faithfully.

    \item \texttt{LLP} represents an order computed by the Layered Label Propagation
    algorithm~\cite{BRSV11}.

    \item \texttt{Spectral} order is given by the second smallest
    eigenvector of the Laplacian matrix of the graph~\cite{JM92}.

    \item \texttt{Multiscale} is an algorithm based on the multi-level algebraic methodology
    suggested for solving \MLogA~\cite{ST11}.

    \item \texttt{SlashBurn} is a method for matrix reordering~\cite{YKF14}.
\end{compactitem}

\subsection{Effect of BP parameters}
\texttt{BP} has a number of parameters that can affect its quality and performance.
In the following we discuss some of the parameters and explain our choice of their
default values.

An important aspect of \texttt{BP} is how two sets, $V_1$ and
$V_2$, are initialized in Algorithm~\ref{algo:bp}. Arguably the initialization
procedure might affect the quality of the final vertex order. To verify the hypothesis,
we implemented four initialization techniques that bisect a given graph:
\texttt{Random}, \texttt{Natural}, \texttt{BFS}, and \texttt{Minhash}.
The techniques order the data vertices, $\D$, using the corresponding
algorithm, and then split the order into two sets of equal size.
In the experiment, we intentionally consider only the simplest and most efficient
bisection techniques so as to keep the complexity of the \texttt{BP} algorithm low.
Figure~\ref{fig:init} illustrates the effect
of the initialization methods for graph bisection.
Note that initialization plays a role to some extent, and there is no consistent
winner. \texttt{BFS} is the best initialization for three of the graphs but
does not produce an improvement on the indexes. One explanation is that the indexes
contain high-degree query vertices, that make the \texttt{BFS} order
essentially random. Overall, the difference between the final
results is not substantial, and even the worst initialization yields
better orders than the alternative algorithms do.
Therefore, we utilize the simplest approach, \texttt{Random}, for graphs
and \texttt{Minhash} for indexes as the default technique
for bisection initialization.

\begin{figure}[htbp]
    \begin{minipage}[b]{0.48\textwidth}
    \centering
    \includegraphics[width=0.99\textwidth]{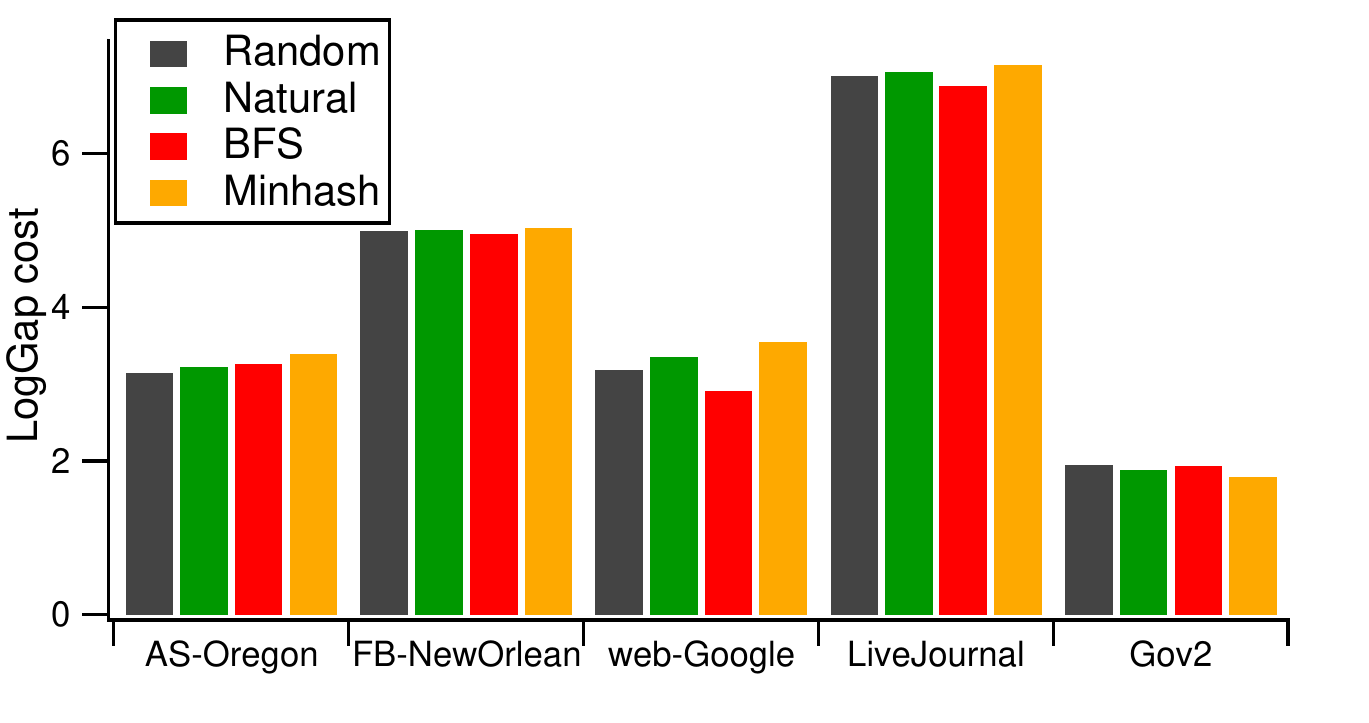}
    \caption{$\LogGap$ cost of the resulting order produced with different
        initialization approaches for graph bisection.}
    \label{fig:init}
    \end{minipage}
    \hfill
    \begin{minipage}[b]{0.48\textwidth}
    \centering
    \includegraphics[width=0.99\textwidth]{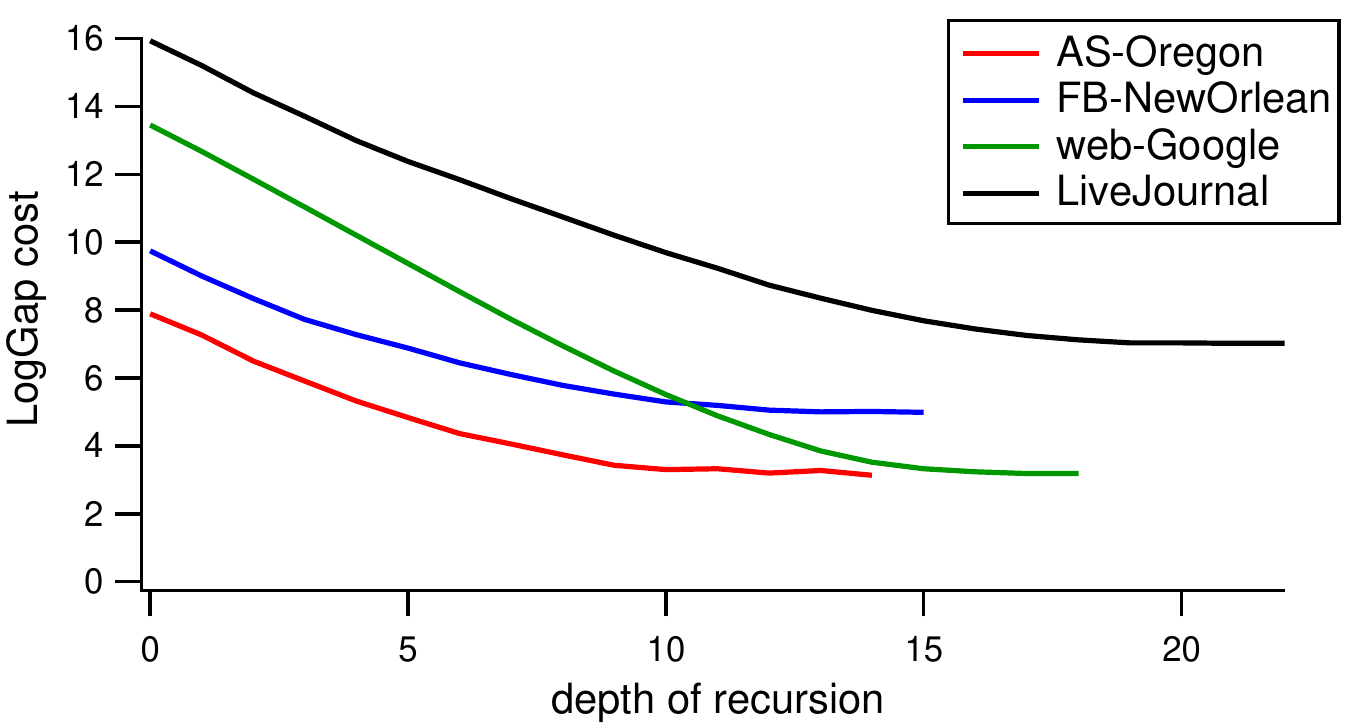}
    \caption{$\LogGap$ cost of the resulting order produced with a
        fixed depth of recursion. Note that the last few splits
        make insignificant contributions to the final quality.}
    \label{fig:depth}
    \end{minipage}
\end{figure}

Is it always necessary to perform $\log n$ levels of recursion to get a
reasonable solution? Figure~\ref{fig:depth} shows the quality of the
resulting vertex order after a fixed number, $i$, of recursion splits.
For every $i$ (that is, when there are $2^i$ disjoint sets), we stop the
algorithm and measure the quality of the order induced by the
random assignment of vertices respecting the partition. It turns out
that graph bisection is beneficial only when $\D$ contains more than
a few tens of vertices. In our implementation, we set $(\log n - 5)$ for the
depth of recursion, which slightly reduces the overall running time.
It might be possible to improve
the final quality by finding an optimal solution (e.g., using an exhaustive search or
a linear program) for small subgraphs on the lowest levels of the recursion.
We leave the investigation for future research.

Figure~\ref{fig:iters} illustrates the speed of convergence of our optimization procedure
utilized for improving graph partitioning in Algorithm~\ref{algo:bp}.
The two sets approach a locally optimal state within a few iterations.
The number of required iterations increases, as the depth of recursion gets larger.
Generally, the number of moved vertices per iteration does not exceed $1\%$ after $20$
iterations, even for the deepest recursion levels.
Therefore, we use $20$ as the default number of iterations in all our experiments.

\begin{figure}
    \centering
    \includegraphics[width=0.49\textwidth]{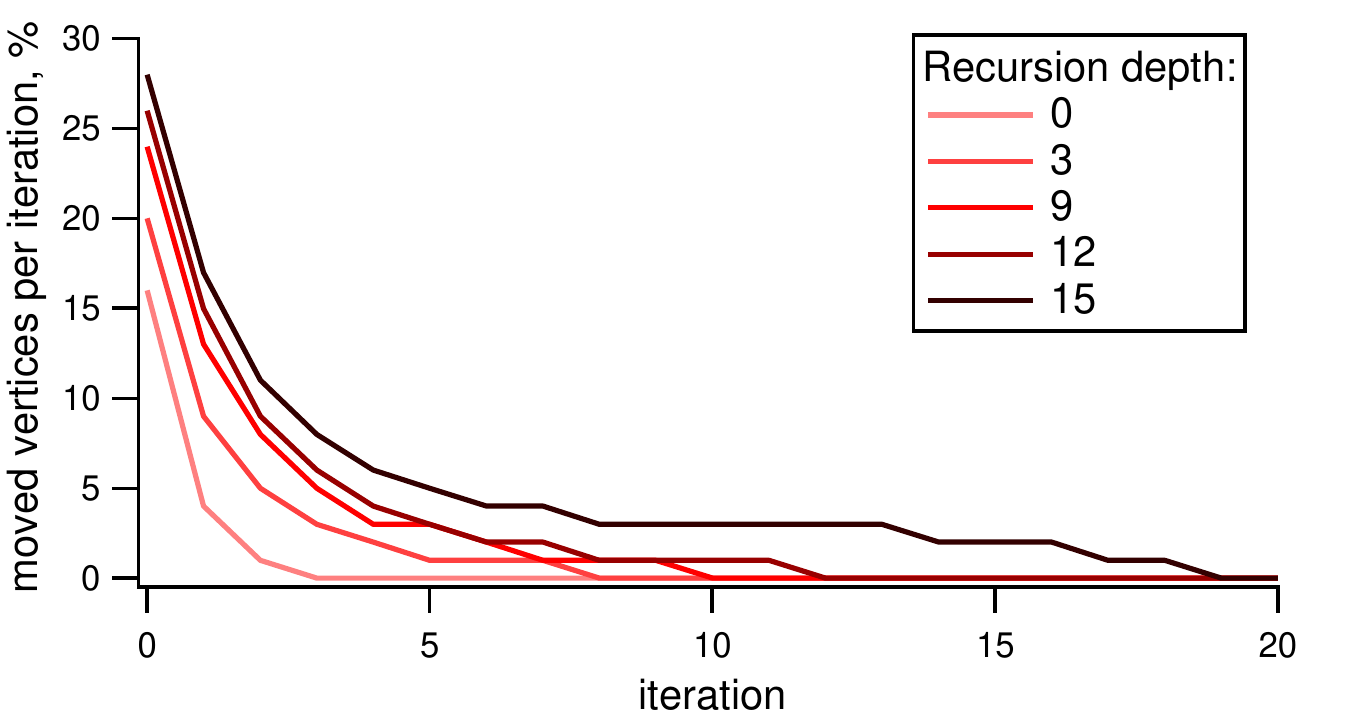}
    \caption{The average percentage of moved vertices on an iteration of
        Algorithm~\ref{algo:bp} for various levels of recursion. The data
        is computed for \texttt{LiveJournal}.}
    \label{fig:iters}
\end{figure}


\begin{figure}[h!]
    \begin{minipage}[b]{0.98\textwidth}
        \begin{subfigure}[t]{.15\textwidth}
            \centering
            \includegraphics[width=\textwidth]{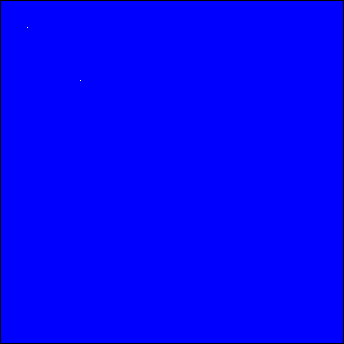}
            \caption{\texttt{Natural}}
        \end{subfigure}
        \hfill
        \begin{subfigure}[t]{.15\textwidth}
            \centering
            \includegraphics[width=\textwidth]{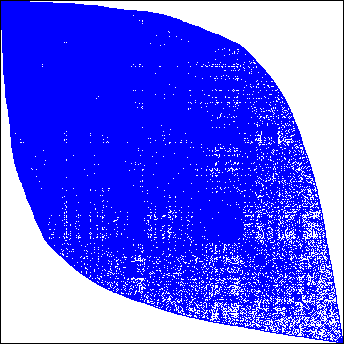}
            \caption{\texttt{BFS}}
        \end{subfigure}
        \hfill
        \begin{subfigure}[t]{.15\textwidth}
            \centering
            \includegraphics[width=\textwidth]{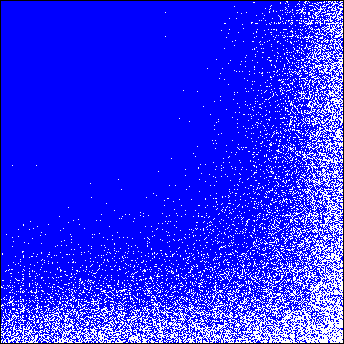}
            \caption{\texttt{Minhash}}
        \end{subfigure}
        \hfill
        \begin{subfigure}[t]{.15\textwidth}
            \centering
            \includegraphics[width=\textwidth]{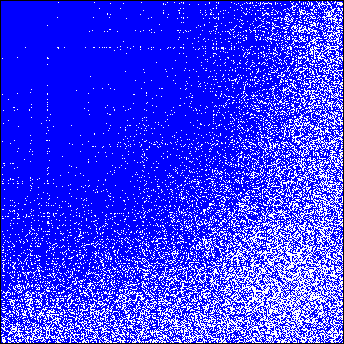}
            \caption{\texttt{TSP}}
        \end{subfigure}
        \hfill
        \begin{subfigure}[t]{.15\textwidth}
            \centering
            \includegraphics[width=\textwidth]{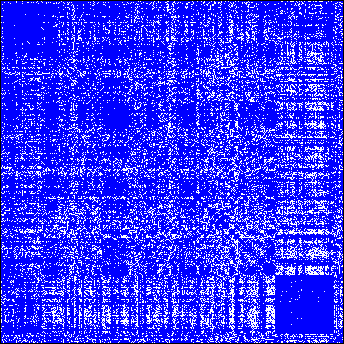}
            \caption{\texttt{LLP}}
        \end{subfigure}
        \hfill
        \\
        \begin{subfigure}[t]{.15\textwidth}
            \centering
            \includegraphics[width=\textwidth]{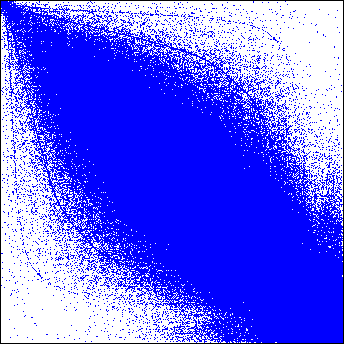}
            \caption{\texttt{Spectral}}
        \end{subfigure}
        \hfill
        \begin{subfigure}[t]{.15\textwidth}
            \centering
            \includegraphics[width=\textwidth]{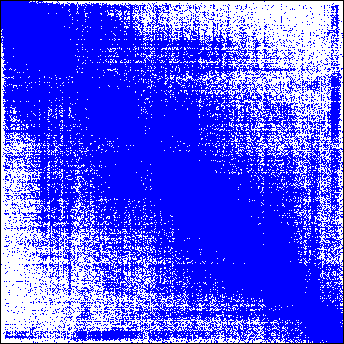}
            \caption{\texttt{Multiscale}}
        \end{subfigure}
        \hfill
        \begin{subfigure}[t]{.15\textwidth}
            \centering
            \includegraphics[width=\textwidth]{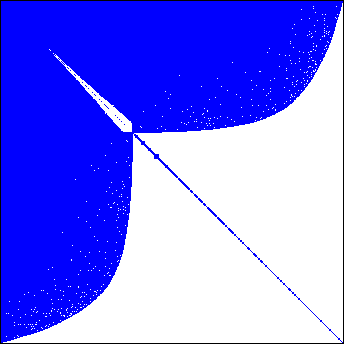}
            \caption{\texttt{SlashBurn}}
        \end{subfigure}
        \hfill
        \begin{subfigure}[t]{.15\textwidth}
            \centering
            \includegraphics[width=\textwidth]{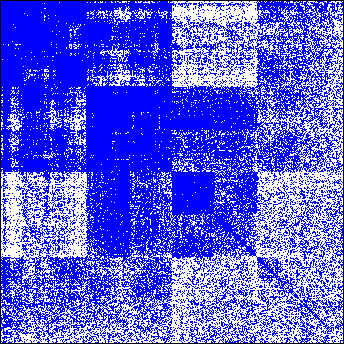}
            \caption{\texttt{BP}}
        \end{subfigure}
        \caption{Adjacency matrices of \texttt{FB-NewOrlean} after applying
            various reordering algorithms; nonzero elements are blue.}
        \label{fig:heatmap}
    \end{minipage}
\end{figure}

\subsection{Compression ratio}

Table~\ref{table:or-graph} presents a comparison of various reordering
methods on social networks and web graphs. We evaluate the following measures:
(i)~the cost of the \BiMLogA problem ($\LogGap$),
(ii)~the cost of the \MLogA problem (the logarithmic difference averaged over the edges, $\Log$),
(iii)~the average number of bits per edge needed to encode the graph with WebGraph~\cite{BV04} (referred to as \BV).
The results suggest that \texttt{BP} yields the best compression on
all but one instance, providing an $5-20\%$ improvement over the best alternative.
An average gain over a non-reordered solution reaches impressive $50\%$.
The runner-up approaches, \texttt{TSP}, \texttt{LLP}, and \texttt{Multiscale},
also significantly outperform the natural order. However, their straightforward
implementations are not scalable for large graphs (none of them is able to process \texttt{Twitter}
within a few hours), while efficient implementations are arguably more complicated than \texttt{BP}.

The computed results for \texttt{FB-300M} and \texttt{FB-1B} demonstrate that
the new reordering technique is beneficial for very large graphs, too.
Unfortunately, we were not able to calculate the compression rate for the
graphs, as WebGraph~\cite{BV04} does not provide distributed implementation.
However, the experiment indicates that \texttt{BP} outperforms
\texttt{Natural} by around $50\%$ and outperforms \texttt{Minhash} by
around $30\%$.

\begin{table}
\vspace{-0.85cm}
\small
    \centering
    \begin{tabular}{llrrr}\toprule
      \multicolumn{1}{c}{Graph} & \multicolumn{1}{c}{Algorithm} %
      & \multicolumn{1}{c}{$\LogGap$} & \multicolumn{1}{c}{$\Log$} & \multicolumn{1}{c}{\BV} \\
      \midrule
            \texttt{Enron}
            & \texttt{Natural} & $5.01$ & $9.82$ & $7.80$ \\
            & \texttt{BFS} & $4.86$ & $9.97$ & $7.70$ \\
            & \texttt{Minhash} & $4.91$ & $10.12$ & $7.68$ \\
            & \texttt{TSP} & $3.95$ & $9.46$ & $6.58$ \\
            & \texttt{LLP} & $3.96$ & $8.55$ & $6.51$ \\
            & \texttt{Spectral} & $5.43$ & $9.41$ & $8.60$ \\
            & \texttt{Multiscale} & $4.23$ & $\boldsymbol{8.00}$ & $6.90$\\
            & \texttt{SlashBurn} & $5.11$ & $10.18$ & $8.05$ \\
            & \texttt{BP} & $\boldsymbol{3.69}$ & $8.26$ & $\boldsymbol{6.24}$\\
        \midrule
            \texttt{AS-Oregon}
            & \texttt{Natural} & $7.88$ & $12.06$ & $13.34$\\
            & \texttt{BFS} & $4.71$ & $11.06$ & $7.97$ \\
            & \texttt{Minhash} & $4.47$ & $11.17$ & $7.56$ \\
            & \texttt{TSP} & $3.59$ & $10.39$ & $6.66$\\
            & \texttt{LLP} & $4.42$ & $8.32$ & $7.47$ \\
            & \texttt{Spectral} & $5.64$ & $9.53$ & $8.76$ \\
            & \texttt{Multiscale} & $4.53$ & $\boldsymbol{7.23}$ & $7.31$\\
            & \texttt{SlashBurn} & $4.50$ & $10.66$ & $8.74$ \\
            & \texttt{BP} & $\boldsymbol{3.15}$ & $9.21$ & $\boldsymbol{6.25}$\\
        \midrule
            \texttt{FB-NewOrlean}
            & \texttt{Natural} & $9.74$ & $14.29$ & $14.64$ \\
            & \texttt{BFS} & $7.16$ & $12.63$ & $10.79$ \\
            & \texttt{Minhash} & $7.06$ & $12.57$ & $10.62$ \\
            & \texttt{TSP} & $5.62$ & $11.61$ & $8.96$ \\
            & \texttt{LLP} & $5.37$ & $\boldsymbol{9.41}$ & $8.54$\\
            & \texttt{Spectral} & $7.64$ & $11.49$ & $11.79$ \\
            & \texttt{Multiscale} & $5.90$ & $9.58$ & $9.25$ \\
            & \texttt{SlashBurn} & $8.37$ & $13.06$ & $12.65$\\
            & \texttt{BP} & $\boldsymbol{4.99}$ & $9.45$ & $\boldsymbol{8.16}$ \\
        \midrule
            \texttt{web-Google}
            & \texttt{Natural} & $13.39$ & $16.74$ & $20.08$\\
            & \texttt{BFS} & $5.57$ & $11.21$ & $7.69$\\
            & \texttt{Minhash} & $5.65$ & $13.14$ & $6.87$ \\
            & \texttt{TSP} & $3.28$ & $7.99$ & $4.77$ \\
            & \texttt{LLP} & $3.75$ & $6.70$ & $5.13$ \\
            & \texttt{Spectral} & $6.68$ & $10.25$ & $9.16$ \\
            & \texttt{Multiscale} & $\boldsymbol{2.72}$ & $\boldsymbol{4.82}$ & $\boldsymbol{4.10}$ \\
            & \texttt{SlashBurn} & $8.02$ & $14.46$ & $10.29$ \\
            & \texttt{BP} & $3.17$ & $7.74$ & $4.68$ \\
      \midrule
          \texttt{LiveJournal}
          & \texttt{Natural} & $10.43$ & $17.44$ & $14.61$ \\
          & \texttt{BFS} & $10.52$ & $17.59$ & $14.69$ \\
          & \texttt{Minhash} & $10.79$ & $17.76$ & $15.07$ \\
          & \texttt{LLP} & $7.46$ & $\boldsymbol{12.25}$ & $11.12$ \\
          & \texttt{BP} & $\boldsymbol{7.03}$ & $12.79$ & $\boldsymbol{10.73}$ \\
      \midrule
          \texttt{Twitter}
          & \texttt{Natural} & $15.23$ & $23.65$ & $21.56$ \\
          & \texttt{BFS} & $12.87$ & $22.69$ & $17.99$ \\
          & \texttt{Minhash} & $10.43$ & $21.98$ & $14.76$ \\
          & \texttt{BP} & $\boldsymbol{7.91}$ & $\boldsymbol{20.50}$ & $\boldsymbol{11.62}$ \\
      \midrule
          \texttt{FB-300M}
          & \texttt{Natural} & $17.65$ & $25.34$ &  \\
          & \texttt{Minhash} & $13.06$ & $24.9$ &  \\
          & \texttt{BP} & $\boldsymbol{8.39}$ & $\boldsymbol{18.13}$ &  \\
      \midrule
          \texttt{FB-1B}
          & \texttt{Natural} & $19.63$ & $27.22$ &  \\
          & \texttt{Minhash} & $14.60$ & $26.89$ &  \\
          & \texttt{BP} & $\boldsymbol{8.66}$ & $\boldsymbol{18.36}$ &  \\
        \bottomrule
    \end{tabular}
    \caption{Reordering results of various algorithms on graphs:
        the costs of $\MLogA$, $\BiMLogA$, and the number of
        bits per edge required by \BV.
        The best results in every column are highlighted.
    We present the results that completed the computation within a few hours.}
    \label{table:or-graph}
\end{table}

\begin{table}
\vspace{-0.5cm}
\small
    \centering
    \begin{tabular}{llrrr}\toprule
        \multicolumn{1}{c}{Index} & \multicolumn{1}{c}{Algorithm} %
        & \multicolumn{1}{c}{$\LogGap$} & \multicolumn{1}{c}{\PEF} & \multicolumn{1}{c}{\BIC} \\
        \midrule
        \texttt{Gov2}
        & \texttt{Natural} & $2.12$ & $3.12$ & $2.52$ \\
        & \texttt{BFS} & $2.07$ & $3.00$ & $2.44$ \\
        & \texttt{Minhash} & $2.12$ & $3.12$ & $2.52$ \\
        & \texttt{BP} & $\boldsymbol{1.81}$ & $\boldsymbol{2.44}$ & $\boldsymbol{1.95}$ \\
        \midrule
        \texttt{ClueWeb09}
        & \texttt{Natural} & $2.91$ & $4.99$ & $4.05$ \\
        & \texttt{BFS} & $2.91$ & $4.99$ & $4.06$ \\
        & \texttt{Minhash} & $2.91$ & $4.99$ & $4.05$ \\
        & \texttt{BP} & $\boldsymbol{2.55}$ & $\boldsymbol{4.34}$ & $\boldsymbol{3.50}$ \\
        \midrule
        \texttt{FB-Posts-1B}
        & \texttt{Natural} & $8.03$ & $10.19$ & $9.95$ \\
        & \texttt{Minhash} & $3.41$ & $4.96$ & $4.24$ \\
        & \texttt{BP} & $\boldsymbol{2.95}$ & $\boldsymbol{4.18}$ & $\boldsymbol{3.61}$ \\
        \bottomrule
    \end{tabular}
    \caption{Reordering results of various algorithms on inverted indexes with highlighted best
        results.}
    \label{table:or-index}
\end{table}

The compression ratio of inverted indexes is illustrated in Table~\ref{table:or-index},
where we evaluate the Partitioned Elias-Fano~\cite{OV14} encoding and Binary
Interpolative Coding~\cite{moffat00binary} (respectively \PEF and \BIC). Here the
results are reported in average bits per edge. Again, our new algorithm largely
outperforms existing approaches in terms of both $\LogGap$ cost and compression rate.
\texttt{BP} has a large impact on the indexes, achieving a $22\%$ and a $15\%$
compression improvement over alternatives; these gains are almost identical for \PEF and \BIC.

An interesting question is why does the new algorithm perform best on most of the
tested graphs. In Figure~\ref{fig:gaps} we analyze the number of gaps between
consecutive numbers of graph adjacency lists. It turns out that \texttt{BP} and
\texttt{LLP} have quite similar gap distributions, having notably more
shorter gaps than the alternative methods. Note that the number of
edges that the \BV encoding is able to copy is related to the number of
consecutive integers in the adjacency lists; hence short gaps strongly influences its performance.
At the same time, \texttt{BP} is slightly better
at longer gaps, which is a reason why the new algorithm yields a higher
compression ratio.

\begin{figure*}
    \begin{minipage}[b]{0.98\textwidth}
        \begin{subfigure}[t]{.49\textwidth}
            \centering
            \includegraphics[width=\textwidth]{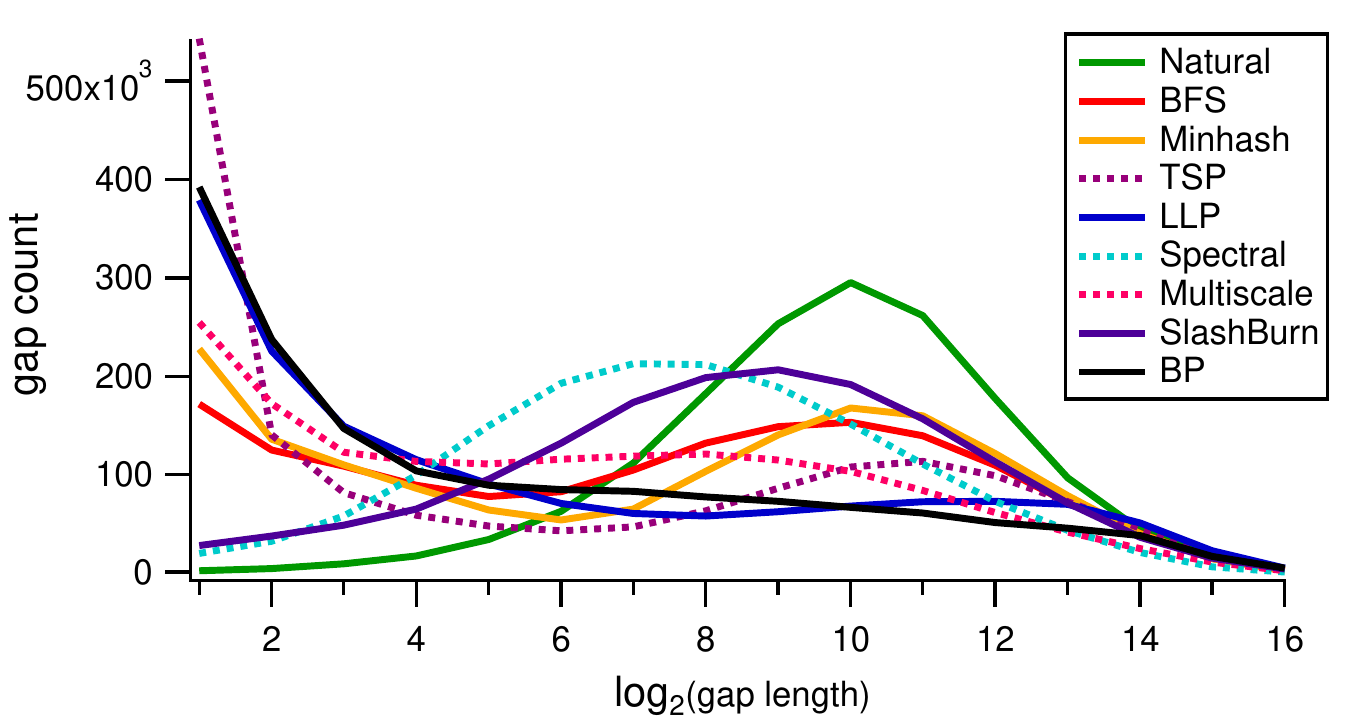}
            \caption{\texttt{FB-NewOrlean}}
        \end{subfigure}
        \hfill
        \begin{subfigure}[t]{.49\textwidth}
            \centering
            \includegraphics[width=\textwidth]{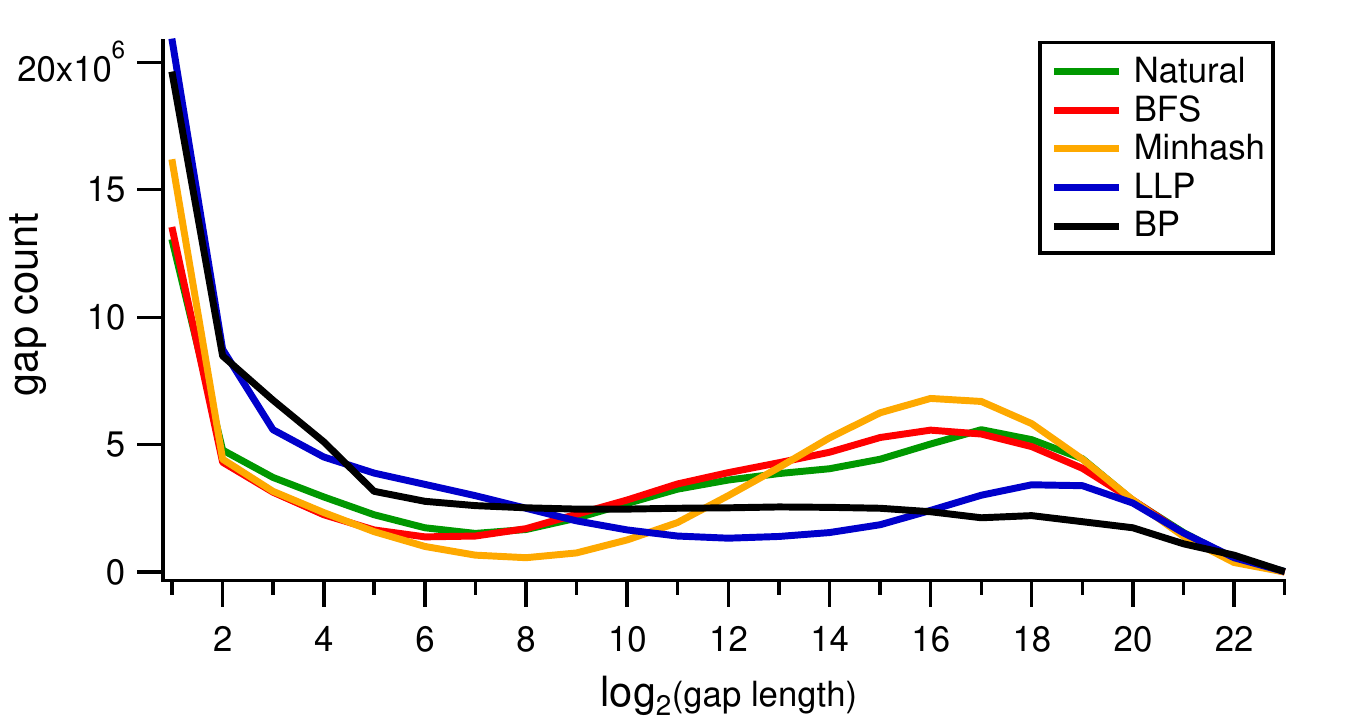}
            \caption{\texttt{LiveJournal}}
        \end{subfigure}
    \end{minipage}
    \caption{Distribution of gaps between consecutive elements of graph adjacency
      lists induced by various algorithms.}
    \label{fig:gaps}
\end{figure*}

We point out that the cost of \BiMLogA, $\LogGap$, is more relevant for the compression rate than
the cost of \MLogA; see Figure~\ref{fig:corr}. The observation agrees with the previous evaluation
of Boldi et al.~\cite{BRSV11} and motivates our research on the former problem.
The Pearson correlation
coefficients between the $\LogGap$ cost and the average number of bit per edge using
\BV, \PEF, and \BIC encoding schemes are $0.9853$, $0.8487$, and $0.8436$, respectively.
While the high correlation between $\LogGap$ and \BV is observed earlier~\cite{BV04,CKLMPR09},
the relation between $\LogGap$ and \PEF or \BIC is a new phenomenon.
A possible explanation is that the schemes encode a sequence of $k$ integers
in the range $[1..n]$ using close to the information-theoretic minimum of
$k(1+\lfloor \log_2(n/k) \rfloor)$ bits~\cite{OV14}, which is equivalent to our optimization
function utilized in Algorithm~\ref{algo:bp}.
It might be possible to construct a better model
for the two encoding schemes, where the cost of the optimization
problem has a higher correlation with the final compression ratio. For example, this can be
achieved by increasing the weights of ``short'' gaps that are generally require more than
$\log(\gap)$ bits. We leave the question for future investigation.

\begin{figure}
    \centering
    \includegraphics[width=0.49\textwidth]{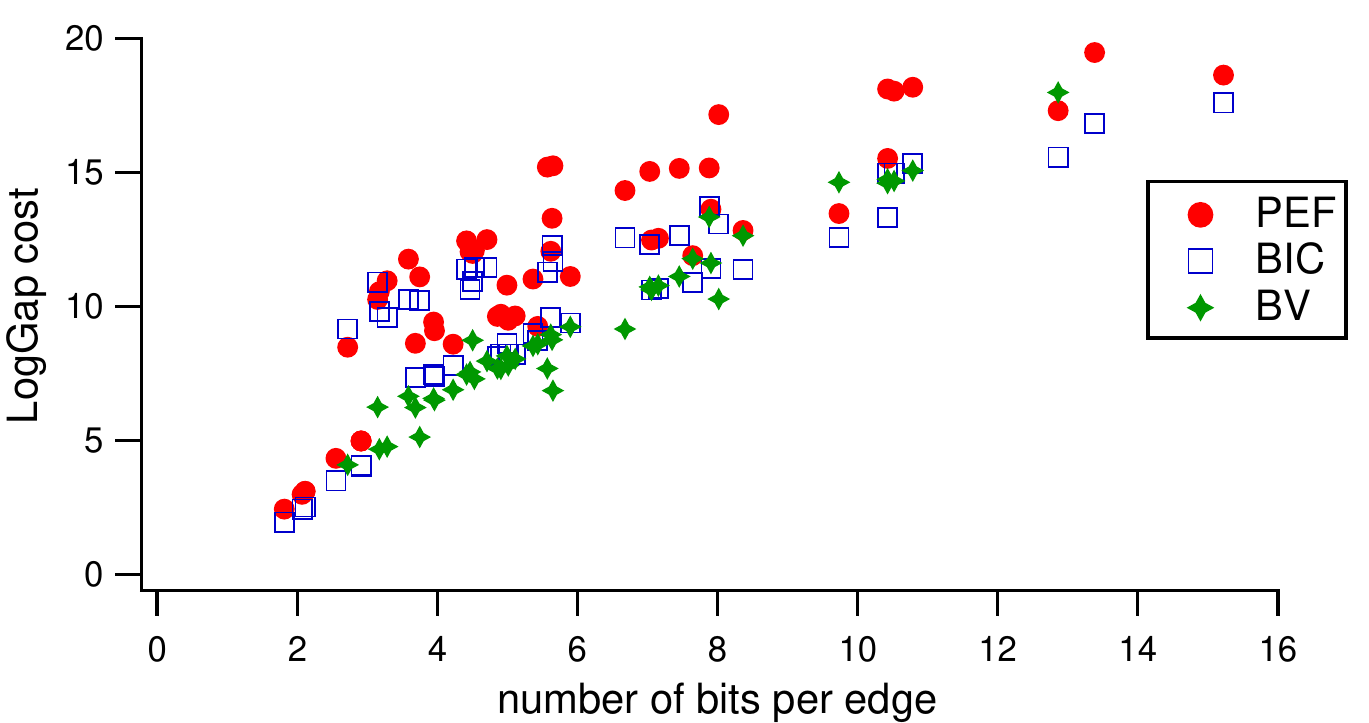}
    \caption{$\LogGap$ cost against the average number of bits per edge using various encoding
        schemes.}
    \label{fig:corr}
\end{figure}

Figure~\ref{fig:heatmap} presents an alternative comparison of the impact of the reordering
algorithms on the \texttt{FB-NewOrlean} graph. Note that only \texttt{BP} and
\texttt{LLP} are able to find communities in the graph (dense subgraphs), that can be compressed
efficiently. The recursive nature of \texttt{BP} is also clearly visible.

\subsection{Running time}
We created and tested two implementations of our algorithm,
parallel and distributed. 
The parallel version is implemented
in C++11 and compiled with the highest optimization settings.
The tests are performed on a machine with
Intel(R) Xeon(R) CPU E5-2660 @ 2.20GHz (32 cores) with 128GB RAM.
Our algorithm is highly scalable; the largest instances of our dataset,
\texttt{Gov2}, \texttt{ClueWeb09} and \texttt{FB-Posts-1B}, are processed
with \texttt{BP} within $29$, $129$, and $163$ minutes, respectively.
In contrast, even the simplest \texttt{Minhash}
takes $14$, $42$, and $70$ minutes for the indexes. \texttt{Natural} and \texttt{BFS}
also have comparable running times on the graphs.
Our largest graphs, \texttt{Twitter} and \texttt{LiveJournal}, require
$149$ and $3$ minutes; all the smaller graphs are processed within a few
seconds. In comparison, the author's implementation of \texttt{LLP} with the
default settings takes $23$ minutes on \texttt{LiveJournal} and is not able to
process \texttt{Twitter} within a reasonable time.
The other alternative methods are less efficient; for instance,
\texttt{Multiscale} runs $12$ minutes and
\texttt{TSP} runs $3$ minutes on \texttt{web-Google}.
The single-machine implementation of \texttt{BP} is also memory-efficient,
utilizing less than twice the space required to store the graph edges.

The distributed version of \texttt{BP} is implemented in Java. We run experiments using the
distributed implementation only on \texttt{FB-300M} and \texttt{FB-1B} graphs, using
a cluster of a few tens of machines. \texttt{FB-300M} is processed
within $350$ machine-hours, while the computation on \texttt{FB-1B} takes
around $2800$ machine-hours. In comparison, the running time of the \texttt{Minhash}
algorithm is $20$ and $60$ machine-hours on the same cluster configuration,
respectively. Despite the fact that our implementation is a part of a
general graph partitioning framework~\cite{SKKSP16}, which is not specifically optimized for the problem,
\texttt{BP} scales almost linearly with the size of the utilized cluster and processes huge graphs within a few hours.




\section{Conclusions and Future Work}
\label{sect:future}
We presented a new theoretically sound algorithm for graph reordering problem
and experimentally proved that the resulting vertex orders allow to
compress graphs and indexes more efficiently than the existing
approaches. The method is highly scalable, which is demonstrated via
evaluation on several graphs with billions of vertices and edges.
While we see impressive gains in the compression
ratio, we believe there is still much room for further improvement.
In particular, our graph bisection technique ignore the
freedom of \df{orienting} the decomposition tree. An interesting question
is whether a postprocessing step that ``flips'' left and right
children of tree nodes can be helpful. It is shown in~\cite{BEFN01}
that there is an $\Oh(n^{2.2})$-time algorithm that computes an optimal
tree orientation for the \MLA problem. Whether there exists a similar
algorithm for \MLogA or \BiMLogA, is open.

While our primary motivation is compression, graph reordering plays an important
role in a number of applications. In particular,
various graph traversal algorithms can be accelerated if the in-memory graph layout
takes advantage of the cache architecture. Improving vertex and edge locality is
important for fast node/link access operations, and thus can be beneficial
for generic graph algorithms and applications~\cite{SDB15}. 
We are currently working on exploring this area and investigating how
reordering of graph vertices can improve cache and memory utilization.

From the theoretical point of view, it is interesting to devise better
approximation algorithms for the \MLogA and \BiMLogA problems.

\section{Acknowledgments}
We thank Yaroslav Akhremtsev, Mayank Pundir, and Arun Sharma for fruitful discussions of the problem.
We thank Ilya Safro for the help with running experiments with the \texttt{Multiscale} method.

\balance
\bibliographystyle{abbrv}
\bibliography{paper}

\begin{thebibliography}{10}

\bibitem{SKKSP16}
A, B.~Shalita, I.~K. Karrer, A.~Sharma, A.~Presta, A.~Adcock, H.~Kllapi, and
  M.~Stumm.
\newblock Social hash: An assignment framework for optimizing distributed
  systems operations on social networks.
\newblock In {\em Networked Systems Design and Implementation}, 2016.

\bibitem{AD09}
A.~Apostolico and G.~Drovandi.
\newblock Graph compression by {BFS}.
\newblock {\em Algorithms}, 2(3):1031--1044, 2009.

\bibitem{ARV09}
S.~Arora, S.~Rao, and U.~Vazirani.
\newblock Expander flows, geometric embeddings and graph partitioning.
\newblock {\em Journal of the ACM}, 56(2):5, 2009.

\bibitem{BEFN01}
R.~Bar-Yehuda, G.~Even, J.~Feldman, and J.~Naor.
\newblock Computing an optimal orientation of a balanced decomposition tree for
  linear arrangement problems.
\newblock {\em JGAA}, 5(4):1--27, 2001.

\bibitem{YE10}
Y.~Ben-Haim and E.~Tom-Tov.
\newblock A streaming parallel decision tree algorithm.
\newblock {\em The Journal of Machine Learning Research}, 11:849--872, 2010.

\bibitem{BB05}
R.~Blanco and {\'A}.~Barreiro.
\newblock Document identifier reassignment through dimensionality reduction.
\newblock In {\em Adv. Inf. Retr.}, pages 375--387. 2005.

\bibitem{BB02}
D.~Blandford and G.~Blelloch.
\newblock Index compression through document reordering.
\newblock In {\em Data Compression Conference}, pages 342--351, 2002.

\bibitem{BRSV11}
P.~Boldi, M.~Rosa, M.~Santini, and S.~Vigna.
\newblock Layered label propagation: A multiresolution coordinate-free ordering
  for compressing social networks.
\newblock In {\em World Wide Web}, pages 587--596, 2011.

\bibitem{BV04}
P.~Boldi and S.~Vigna.
\newblock The {W}eb{G}raph framework {I}: {C}ompression techniques.
\newblock In {\em World Wide Web}, pages 595--602, 2004.

\bibitem{Bro97}
A.~Z. Broder.
\newblock On the resemblance and containment of documents.
\newblock In {\em Compression and Complexity of Sequences}, pages 21--29, 1997.

\bibitem{CHKR10}
M.~Charikar, M.~T. Hajiaghayi, H.~Karloff, and S.~Rao.
\newblock $l^2_2$ spreading metrics for vertex ordering problems.
\newblock {\em Algorithmica}, 56(4):577--604, 2010.

\bibitem{CMM07}
M.~Charikar, K.~Makarychev, and Y.~Makarychev.
\newblock A divide and conquer algorithm for d-dimensional arrangement.
\newblock In {\em Symposium on Discrete Algorithms}, pages 541--546, 2007.

\bibitem{CKLMPR09}
F.~Chierichetti, R.~Kumar, S.~Lattanzi, M.~Mitzenmacher, A.~Panconesi, and
  P.~Raghavan.
\newblock On compressing social networks.
\newblock In {\em Knowledge Discovery and Data Mining}, pages 219--228, 2009.

\bibitem{DKSV06}
N.~R. Devanur, S.~A. Khot, R.~Saket, and N.~K. Vishnoi.
\newblock Integrality gaps for sparsest cut and minimum linear arrangement
  problems.
\newblock In {\em Symposium on Theory of Computing}, pages 537--546, 2006.

\bibitem{DAS10}
S.~Ding, J.~Attenberg, and T.~Suel.
\newblock Scalable techniques for document identifier assignment in inverted
  indexes.
\newblock In {\em World Wide Web}, pages 311--320, 2010.

\bibitem{ENRS00}
G.~Even, J.~S. Naor, S.~Rao, and B.~Schieber.
\newblock Divide-and-conquer approximation algorithms via spreading metrics.
\newblock {\em Journal of the ACM}, 47(4):585--616, 2000.

\bibitem{FM82}
C.~M. Fiduccia and R.~M. Mattheyses.
\newblock A linear-time heuristic for improving network partitions.
\newblock In {\em Design Automation}, pages 175--181, 1982.

\bibitem{GJ79}
M.~R. Garey and D.~S. Johnson.
\newblock {\em Computers and Intractability: A Guide to the Theory of
  NP-Completeness}.
\newblock W. H. Freeman \& Co., 1979.

\bibitem{Han89}
M.~D. Hansen.
\newblock Approximation algorithms for geometric embeddings in the plane with
  applications to parallel processing problems.
\newblock In {\em Foundations of Computer Science}, pages 604--609, 1989.

\bibitem{JM92}
M.~Juvan and B.~Mohar.
\newblock Optimal linear labelings and eigenvalues of graphs.
\newblock {\em Discrete Applied Mathematics}, 36(2):153--168, 1992.

\bibitem{KL70}
B.~W. Kernighan and S.~Lin.
\newblock An efficient heuristic procedure for partitioning graphs.
\newblock {\em Bell System Technical Journal}, 49(2):291--307, 1970.

\bibitem{KCHM10}
H.~Kwak, C.~Lee, H.~Park, and S.~Moon.
\newblock What is {T}witter, a social network or a news media?
\newblock In {\em World Wide Web}, pages 591--600, 2010.

\bibitem{YKF14}
Y.~Lim, U.~Kang, and C.~Faloutsos.
\newblock {SlashBurn}: Graph compression and mining beyond caveman communities.
\newblock {\em IEEE Transactions on Knowledge and Data Engineering},
  26(12):3077--3089, 2014.

\bibitem{MP15}
S.~Maneth and F.~Peternek.
\newblock A survey on methods and systems for graph compression.
\newblock {\em arXiv preprint arXiv:1504.00616}, 2015.

\bibitem{moffat00binary}
A.~Moffat and L.~Stuiver.
\newblock Binary interpolative coding for effective index compression.
\newblock {\em Information Retrieval}, 3(1), 2000.

\bibitem{OV14}
G.~Ottaviano and R.~Venturini.
\newblock Partitioned {E}lias-{F}ano indexes.
\newblock In {\em SIGIR}, pages 273--282, 2014.

\bibitem{Pet13}
J.~Petit.
\newblock Addenda to the survey of layout problems.
\newblock {\em Bulletin of EATCS}, 3(105), 2013.

\bibitem{RAK07}
U.~N. Raghavan, R.~Albert, and S.~Kumara.
\newblock Near linear time algorithm to detect community structures in
  large-scale networks.
\newblock {\em Physical Review E}, 76(3):036106, 2007.

\bibitem{RSWW02}
K.~H. Randall, R.~Stata, R.~G. Wickremesinghe, and J.~L. Wiener.
\newblock The link database: Fast access to graphs of the web.
\newblock In {\em Data Compression Conference}, pages 122--131, 2002.

\bibitem{RR98}
S.~Rao and A.~W. Richa.
\newblock New approximation techniques for some ordering problems.
\newblock In {\em Symposium on Discrete Algorithms}, pages 211--219, 1998.

\bibitem{ST11}
I.~Safro and B.~Temkin.
\newblock Multiscale approach for the network compression-friendly ordering.
\newblock {\em Journal of Discrete Algorithms}, 9(2):190--202, 2011.

\bibitem{SCSC03}
W.-Y. Shieh, T.-F. Chen, J.~J.-J. Shann, and C.-P. Chung.
\newblock Inverted file compression through document identifier reassignment.
\newblock {\em Information Processing \& Management}, 39(1):117--131, 2003.

\bibitem{SDB15}
J.~Shun, L.~Dhulipala, and G.~E. Blelloch.
\newblock Smaller and faster: {P}arallel processing of compressed graphs with
  {L}igra+.
\newblock In {\em Data Compression Conference}, pages 403--412, 2015.

\bibitem{Sil07}
F.~Silvestri.
\newblock Sorting out the document identifier assignment problem.
\newblock In {\em European Conference on IR Research}, pages 101--112.
  Springer, 2007.

\bibitem{UB13}
J.~Ugander and L.~Backstrom.
\newblock Balanced label propagation for partitioning massive graphs.
\newblock In {\em Web Search and Data Mining}, pages 507--516, 2013.

\bibitem{VMCG09}
B.~Viswanath, A.~Mislove, M.~Cha, and K.~P. Gummadi.
\newblock On the evolution of user interaction in {F}acebook.
\newblock In {\em Workshop on Social Networks}, 2009.

\bibitem{WMB99}
I.~H. Witten, A.~Moffat, and T.~C. Bell.
\newblock {\em Managing gigabytes: compressing and indexing documents and
  images}.
\newblock Morgan Kaufmann, 1999.

\end{thebibliography}

\end{document}